\documentclass[envcountsame,runningheads]{llncs}

\usepackage{graphicx} 
\usepackage{amsmath, amssymb}
\usepackage{booktabs}
\usepackage{xcolor}
\usepackage{tikz}
\usepackage{placeins}
\usepackage{float}
\usepackage{multirow}
\usetikzlibrary{trees,positioning,shapes,arrows.meta}
\usepackage[normalem]{ulem}
\usepackage[T1]{fontenc}
\usepackage{lmodern}
\usepackage{comment}
\usepackage{tcolorbox}
\usepackage{algorithm}
\usepackage[noend]{algpseudocode}

\spnewtheorem{thrm}{Theorem}{\bfseries}{\itshape}
\spnewtheorem{lemm}{Lemma}{\bfseries}{\itshape}
\spnewtheorem{defn}{Definition}{\bfseries}{\itshape}

\renewcommand{\qed}{\hfill$\square$}

\newcommand{\itbf}[1]{\textit{\textbf{#1}}}

\usepackage{tcolorbox}
\tcbuselibrary{skins}
\usepackage{tcolorbox}
\tcbuselibrary{skins, breakable}

\newtcolorbox{soutblock}{
  blanker,
  breakable,
  underlay={
    \draw[red, line width=1.5pt] (interior.north west) -- (interior.south east);
  }
}

\begin{document}
\title{Boyer–Moore Variants for Indeterminate String Matching and Experimental Evaluation}
\titlerunning{Boyer–Moore Variants for Indeterminate Strings}
\author{Neerja Mhaskar\thanks{Corresponding author. Supported by the Natural Sciences \& Engineering Research Council of Canada [Grant Number RGPIN-2024-06915].}\orcidID{0000-0002-3233-6540} \and Nivetha Raj Pappuraj}
\institute{Department of Computing and Software, McMaster University,\\
  1280 Main Street West, Hamilton, Ontario, L8S 4L8, Canada\\
  \email{\{poplin, pappuran\}@mcmaster.ca}
}
\maketitle

\begin{abstract}
We study exact pattern matching on \emph{indeterminate strings}, where a text or pattern position may represent a set of symbols rather than a single letter. Focusing on Boyer--Moore-style methods, we present new bad-character rules (BC Rules~I--IV) and a new good-suffix procedure, computed by \textsc{Fast\_\allowbreak GSR\_\allowbreak Indet\_\allowbreak Shift}, which avoids the per-alignment recomputation used in \textsc{BM\_Indet}~\cite{Neerja2024} by shifting with a single preprocessed position-indexed table. We conduct a systematic experimental evaluation of sixteen algorithms, including classical bad-character adaptations (e.g., Horspool, Sunday, and Zhu--Takaoka) and hybrids that combine these bad-character rules with \textsc{Fast\_\allowbreak GSR\_\allowbreak Indet\_\allowbreak Shift}. Across synthetic scaling experiments and a case study on the \textit{E.\,coli} K-12 MG1655 genome, the \textsc{Fast\_BM\_Indet} hybrids consistently outperform \textsc{BM\_Indet} and \textsc{KMP\_Indet}~\cite{Neerja2024}, in some settings by up to two orders of magnitude. We also find that Zhu--Takaoka is the strongest bad-character-only adaptation on small alphabets and genomic data, while the \textsc{Fast\_BM\_Indet} variant using BC Rule~I offers comparable performance, making it attractive for larger alphabets. We conclude with practical guidance on choosing among these variants for indeterminate-string applications.
\end{abstract}


\section{Introduction}
\label{sec:introduction}

A string is typically modeled as a sequence of letters over a finite alphabet $\Sigma$. Given a text $t$ and pattern $p$ over $\Sigma$, exact pattern matching asks for all positions where $p$ occurs in $t$, and has been studied extensively since the 1970s~\cite{BaezaYates1992,Boyer1977,Knuth1977,Navarro2002}, including the Knuth--Morris--Pratt (KMP)~\cite{Knuth1977} and Boyer--Moore (BM)~\cite{Boyer1977} algorithms. In practice, however, a position may be ambiguous; for instance, a DNA sequence over $\Sigma_{\text{DNA}} = \{A,C,G,T\}$ may allow either $A$ or $C$ at a site. Such uncertainty is captured by \textit{indeterminate} (or \textit{degenerate}) strings, where at least one position contains a subset of $\Sigma$ of size greater than one.

Indeterminate strings arise naturally in genomics: a consensus sequence from aligned reads is indeterminate where multiple bases occur at comparable frequency~\cite{Stormo2000}, and a SNP-annotated reference genome is indeterminate at sites where individuals carry different bases~\cite{IUPAC1985}. Efficient matching is therefore directly relevant to genomic sequence analysis, motivating our case study on real genomic data in Section~\ref{sec:experimental-results}.

Pattern matching on indeterminate strings dates to Fischer and Paterson~\cite{FischerPaterson1974} and Abrahamson~\cite{Abrahamson1987}, with later work on algorithmic properties~\cite{Alzamel2020,Crochemore2016,Daykin2019,HolubSmythWang2008}. Holub, Smyth, and Wang~\cite{HolubSmythWang2008} gave a family of algorithms under varying constraints; their most general achieves $O(mn)$ time by switching between a Boyer--Moore style method for solid strings and Shift-Add for indeterminate ones. Iliopoulos and Radoszewski~\cite{Iliopoulos2016} gave an  $\mathcal{O}(2^{\sigma} n \log m)$ time FFT-based algorithm, where $\sigma$ is the alphabet size, but Hasibi et al.~\cite{HMZ2026} removed this exponential dependence with an FFT-based method running in $O\bigl(n\log m + |Occ_{\diamondsuit}|\cdot \min(m, D_1 + D_2)\bigr)$ time, where $|Occ_{\diamondsuit}|$ is the number of occurrences reported after converting to partial-word matching and $D_1$ and $D_2$ are the numbers of indeterminate positions in the pattern and text.

Crochemore et al.~\cite{Crochemore2016} gave an $O(nk)$ suffix-tree algorithm for
\textit{conservative} degenerate matching, where text and pattern hold at most $k$
indeterminate positions, and Daykin et al.~\cite{Daykin2019} proposed DBWT, which
builds the text's Burrows--Wheeler transform in $O(n)$ time and reports occurrences
in $O(km^2+\text{occ})$ time. Gawrychowski, Ghazawi, and Landau~\cite{Gawrychowski2020}
studied equal-length indeterminate matching under Cartesian-tree, order-preserving,
and parameterized relations, giving an $O(n\log^2 n)$-time Cartesian-tree algorithm
and NP-hardness for the other two even when positions hold at most two characters;
the journal version~\cite{Gawrychowski2026} improved the Cartesian-tree bound to
$O(n\log n\log\log n)$ time and $O(n)$ space. Dehghani et al.~\cite{Neerja2024,KMPIndet}
introduced prefix-array-based KMP- and BM-style algorithms for indeterminate strings
and noted that systematically adapting the Boyer--Moore variant family remained
open the direction we pursue here.

In this paper, we first adapt the classical Boyer--Moore bad-character only variants to indeterminate strings. We then propose four new bad-character rules (BC~Rules~I--IV) together with a new good-suffix rule, computed by the \textsc{Fast\_\allowbreak GSR\_\allowbreak Indet\_\allowbreak Shift} algorithm. Combining each bad-character rule with this good-suffix rule yields four new algorithm variants, which we call the \textsc{Fast\_\allowbreak BM\_\allowbreak Indet} variants. We show that, among the bad-character-only
variants for indeterminate strings including Holub's algorithm~\cite{HolubSmythWang2008}, the
Zhu--Takaoka adaptation performs best, whereas among the four
\textsc{Fast\_BM\_Indet} variants, the one pairing BC~Rule~I with the
\textsc{Fast\_\allowbreak GSR\_\allowbreak Indet\_\allowbreak Shift} good-suffix rule performs best. For
completeness, we also show that all four \textsc{Fast\_\allowbreak BM\_\allowbreak Indet} variants outperform the existing\allowbreak\ \textsc{KMP\_\allowbreak Indet} and \textsc{BM\_\allowbreak Indet} algorithms, and perform comparably to the brute-force algorithm.

The remainder of the paper is organized as follows.
Section~\ref{sec:preliminaries} gives definitions and notation, and
Section~\ref{sec:bm} surveys the classical Boyer--Moore algorithm and the variant families considered here. Section~\ref{sec:main-results} presents our indeterminate adaptations of these classical variants and then introduces four new Boyer--Moore style algorithms for indeterminate strings. Finally,  Section~\ref{sec:experimental-results} presents the experimental evaluation of sixteen algorithms.


\section{Preliminaries}
\label{sec:preliminaries}
Let $\Sigma$ be a finite ordered alphabet of size $|\Sigma| =
\sigma$. A \textit{solid string} $x[1..n]$ over $\Sigma$ is a
sequence of $n$ letters each drawn from $\Sigma$, where $x[i]
\in \Sigma$ for all $1 \leq i \leq n$. The \textit{length} of
$x$ is $|x| = n$ and the \textit{empty string} is denoted
$\varepsilon$. For $1 \leq i \leq j \leq n$, $u = x[i..j]$ is
a \textit{substring} of $x$, empty if $j < i$. If $i = 1$,
$u$ is a \textit{prefix} of $x$; if $j = n$, $u$ is a
\textit{suffix} of $x$. The \textit{reverse} of $x[1..m]$ is
the solid string $x^R[1..m]$ defined by $x^R[i] = x[m-i+1]$
for $1 \leq i \leq m$. 
Let $\Sigma'$ denote the set of all non-empty subsets of
$\Sigma$. An element of $\Sigma'$ of cardinality one, is called a \textit{solid letter}, and of cardinality greater than one is called an \textit{indeterminate letter}. Abusing notation for ease of exposition,  we identify a solid letter $\{c\} \in \Sigma'$ with the corresponding character $c \in \Sigma$, and write $c$ for both. An
\textit{indeterminate string} (also called a \textit{degenerate
string}) $T[1..n]$ is a sequence of elements drawn from
$\Sigma'$, with at least one indeterminate letter. 
A position $i$ of an indeterminate string $x$ is called an \textit{indeterminate position}, if $|x[i]| > 1$, otherwise it is called a \textit{solid position}.
Two letters
$\lambda_1, \lambda_2 \in \Sigma'$ \textit{match}, written
$\lambda_1 \approx \lambda_2$, if and only if $\lambda_1 \cap
\lambda_2 \neq \emptyset$; otherwise they \textit{mismatch},
written $\lambda_1 \not\approx \lambda_2$. 

The \textit{prefix array} $\pi_x[1..n]$ of a
string $x[1..n]$ is defined as an integer array where for each position $i$,  $\pi_x[i]$ denotes the length of the longest prefix
of $x[1..n]$ that matches a prefix of the suffix $x[i..n]$. Formally, $\pi_x[i] = \max\{\, \ell \mid 0 \leq \ell \leq n-i+1
\text{ and } x[1..\ell] \approx x[i..i+\ell-1] \,\}$. The prefix array is computed in $O(n)$ time for solid
strings~\cite{HolubSmythWang2008} and in $O(n\sqrt{n})$ time for indeterminate
strings~\cite{Iliopoulos2016}.  
An \textit{alignment} at position $i$, $1 \leq i \leq n-m+1$,
refers to comparing the pattern against the text substring
$T[i..i+m-1]$, which we call the \textit{window}. A
\textit{shift} $s \geq 1$ advances an alignment at position $i$ to position $i+s$.

It is customary to encode indeterminate strings for their processing. In this paper we use the \textit{power-of-two encoding}~\cite{btt2013}:
each letter of the alphabet $\Sigma = \{a_1, a_2,\allowbreak \ldots,\allowbreak a_\sigma\}$ of size
$\sigma$ is mapped to a distinct power of two, $a_j \mapsto 2^{\,j-1}$, so that
$\Sigma$ is mapped to $\Sigma_{N} = \{2^0, 2^1, \ldots, 2^{\sigma-1}\}$. A non-empty
indeterminate letter $\lambda \subseteq \Sigma$ is represented as the sum of the
codes of its elements, $\sum_{a \in \lambda} a$, which, since the codes occupy
distinct bits, is exactly the bitwise OR of those codes. A character $s$ is
\emph{solid} if its value is a power of two (equivalently, exactly one bit is set),
and two characters $s_1, s_2$ match ($s_1 \approx s_2$) if and only if their bitwise
AND is non-zero ($s_1 \ \&\ s_2 \neq 0$), i.e.\ they share at least one letter of
$\Sigma$.

\textit{Throughout this paper, strings are assumed to be over integer alphabets. We use lowercase letters $t$ and $p$ to denote solid strings, and capital letters $T$ and $P$ to denote indeterminate strings.}

\section{Boyer--Moore Algorithm and its Variants}\label{sec:bm}

The Boyer--Moore algorithm~\cite{Boyer1977}, widely regarded as the
fastest practical algorithm for exact string matching on solid strings over large alphabets, determines the shift solely from the pattern using the bad character and good suffix rules by preprocessing the pattern to create lookup tables for each rule in $O(m)$-time. At each alignment $i$, $P[1..m]$ is compared against
the window $T[i..i+m-1]$ from right to left and a shift $s =
\max(\text{BC shift},\, \text{GS shift})$ is applied on a mismatch,
achieving $O(nm)$ worst-case and $O(n/m)$ best-case search time. 

The Boyer--Moore (BM) variants considered here fall into two families.
\textbf{BC-only} variants rely solely on the bad-character rule; they
differ in which text positions, and how many, they examine after a
mismatch when computing the shift. \textbf{BC+GS} variants keep both the
bad-character and good-suffix rules and shift by the larger of the two,
differing only in how the good-suffix rule is realized to guarantee linear worst-case time. The bad-character rule has several forms. In the simplest, the shift
depends on one mismatched character; some variants use a pair of
consecutive characters. The classical variants studied here, summarised
in Table~\ref{tab:taxonomy}, cover three lookup strategies:
single-character (Horspool, Sunday, Raita, Smith, and Tuned
Boyer--Moore), character-pair (Zhu--Takaoka and Berry--Ravindran), and
the full BC+GS combination of the original Boyer--Moore
algorithm~\cite{Boyer1977}. We extend these baselines to indeterminate
text in Section~\ref{sec:main-results}, which also introduces BC
Rules~I--IV and a new good-suffix algorithm, yielding four new BC+GS
variants. We exclude variants that replace a rule with a fundamentally different
mechanism~\cite{ReverseFactor1994,TurboReverseFactor1994,FastSearch2005},
fuse Boyer--Moore with unrelated techniques~\cite{WuManber1994}, or trade
simplicity for additional bookkeeping~\cite{Crochemore1992}.
\begin{lemm}
\label{lem:bm-complexity}
Boyer--Moore variants using only single-character bad-character
rules require $O(m + \sigma)$ preprocessing time, $O(\sigma)$ space,
and $O(nm)$ worst-case search time. The BC-Only character-pair
variants require $O(m + \sigma^2)$ preprocessing time, $O(\sigma^2)$
space, and $O(nm)$ worst-case search time. The classical
Boyer--Moore algorithm, combining both rules, also requires $O(nm)$
worst-case search time, with  $O(m + \sigma)$ preprocessing for the
bad-character table and $O(m)$ additional preprocessing for the
good-suffix tables.
\end{lemm}

\section{Main Results}
\label{sec:main-results}
In this section, we first adapt the seven Boyer--Moore bad-character-only variants of
Table~\ref{tab:taxonomy} for indeterminate stings, then we present four new BC+GS variants for indeterminate strings.

\subsection{Bad Character Rule for Indeterminate Strings}\label{ssec:indetbc}

The table $\tau$ of Equation~\ref{eq:bm-bc-tau} given in Section~\ref{Appendix} can be used to compute the bad character rule shift for indeterminate strings as follows. Suppose the right-to-left scan fails at position $j$, against the text letter $T[i+j-1]$, which serves as the lookup character in the Boyer--Moore algorithm. If this lookup character is solid, the shift is computed as $\max\!\bigl(1,\, j-\tau[T[i+j-1]]\bigr)$, as described in Section~\ref{sec:bm}. However, if it is indeterminate, say $\lambda$, the shift is computed as,
\begin{equation}
\label{eq:bc-indet}
\textsc{Indet-BCS(BM)}=\max\!\bigl(1,\,j - \max\{\, \tau(c) : c \in \lambda \,\}\bigr).
\end{equation}

Observe that, the five Bad Character only variants (Horspool, Sunday, Raita, Smith, and Tuned BM) compute the shift based on different lookup characters. We use these lookup characters with the table $\tau$ along with the strategy used for solid and indeterminate letters to compute the shift as follows. 

\begin{equation}
\textsc{Indet-BCS(Horspool)}=\max\bigl(1,\,m-\max\{\tau(c):c\in T[i+m-1]\}\bigr)
\end{equation}
\begin{equation}
\textsc{Indet-BCS(Sunday)} =
\begin{cases}
& \max\bigl(1,\,(m+1)-\max\{\tau(c):c\in T[i+m]\}\bigr), \\
& 1, \quad \text{when } i+m>n.
\end{cases}
\end{equation}

\begin{equation}
\label{eq:indet-raita}
\textsc{Indet-BCS(Raita)} = \max
\begin{cases}
\max\Bigl(1,\; m - \max\limits_{c \in T[i+m-1]} \tau(c)\Bigr), \\[8pt]
\max\Bigl(1,\; m - \max\limits_{c \in T[i]} \tau(c)\Bigr), \\[8pt]
\max\Bigl(1,\; m - \max\limits_{c \in T[k]} \tau(c)\Bigr), \text{ where } k=i+\lfloor\frac{m}{2}\rfloor
\end{cases}
\end{equation}

\begin{equation}
\label{eq:indet-smith}
\begin{aligned}
\textsc{Indet-BCS(Smith)} = \max
\begin{cases}
\max\Bigl(1,\; m - \max\limits_{c \in T[i+m-1]} \tau(c)\Bigr), \\[10pt]
\max\Bigl(1,\; (m{+}1) - \max\limits_{c \in T[i+m]} \tau(c)\Bigr),
  \quad i+m \le n.
\end{cases}
\end{aligned}
\end{equation}
\begin{equation}
\label{eq:indet-tuned}
\textsc{Indet-BCS(Tuned BM)} =
\max\Bigl(1,\; m - \max\limits_{c \in T[i+m-1]} \tau(c)\Bigr).
\end{equation}

The use of $\tau$ for indeterminate strings increases the cost of a bad character table lookup from $O(1)$ to $O(\sigma)$ time. Thus we get,

\begin{lemm}
\label{lem:bc-single}
The Horspool, Sunday, Raita, Smith, and Tuned Boyer--Moore algorithm adaptations
for indeterminate strings require $O(m + \sigma)$ preprocessing time,
$O(\sigma)$ space, and $O(nm\sigma)$ worst-case search time.
\end{lemm}

We adapt the two character-pair Boyer--Moore variants,
\itbf{Zhu--Takaoka} and \itbf{Berry--Ravindran}, using the same
strategy as for the single-character variants. The only difference is
that $\tau$ is replaced by $\tau_{cp}$, and the single lookup character is
replaced by a pair of lookup characters chosen as each variant
prescribes.

\begin{equation}
\label{eq:bc-zt}
\textsc{Indet-BCS(Zhu--Takaoka)}=\max\!\Bigl(1,\;m-\min_{\substack{c_1\in T[i+m-2]\\ c_2\in T[i+m-1]}}\tau_{cp}[c_1,c_2]\Bigr)
\end{equation}
\begin{equation}
\label{eq:bc-br}
\textsc{Indet-BCS(Berry--Ravindran)} =
\begin{cases}
\max\Bigl(1,\;m-\min\limits_{\substack{c_1\in T[i+m]\\ c_2\in T[i+m+1]}}\tau_{cp}[c_1,c_2]\Bigr), & \text{when } i+m<n,\\[12pt]
\max\Bigl(1,\;m-\min\limits_{\substack{c_1\in T[i+m]\\ c_2\in \Sigma}}\tau_{cp}[c_1,c_2]\Bigr), & \text{when } i+m=n,\\[12pt]
1, & \text{when } i+m>n.
\end{cases}
\end{equation}

Analogous to the use of $\tau$, employing $\tau_{cp}$ for indeterminate strings increases the cost of a bad-character table lookup from $O(1)$ to $O(\sigma)$ time. Thus we get,
\begin{lemm}
\label{lem:bc-pair}
The Zhu--Takaoka and Berry--Ravindran algorithm adaptations for indeterminate
strings require $O(m+\sigma^2)$ preprocessing time,
$O(\sigma^2)$ space, and $O(nm\sigma)$ worst-case search time.
\end{lemm}

\subsection{Faster Boyer--Moore Variants}
\label{ssec:fast-bmindet}
In this section, we improve upon the only existing algorithm under the BC+GS category for indeterminate strings~\cite{Neerja2024}. We propose four BC+GS variants extending \textsc{BM\_Indet}~\cite{Neerja2024}. All four use the same fast good-suffix rule and differ only in their bad-character rule; we describe the bad-character rules first, then the good-suffix rule. In Section~\ref{ssec:indetbc}, using $\tau$ naively increases bad-character lookup to $O(\sigma)$ time. To retain $O(1)$ lookup, we propose the following $\tau$-based rule:

\noindent\textbf{BC Rule~I (nearest solid letter):} We keep the table $\tau$ unchanged and modify only which
character we query. When the current text position $i+j-1$ holds a
solid letter $c$, we apply the ordinary bad-character rule. When
$T[i+j-1]$ is instead an indeterminate letter $\lambda$, we avoid
examining every $c \in \lambda$ and instead fall back on the nearest
solid letter in the unmatched prefix of the current window, if one
exists. Let
\[
  j' = \max\bigl\{\, k : 1 \leq k < j \ \text{and}\ T[i+k-1] \text{ is solid} \,\bigr\}
\]
be the position of that letter within the window, and let
$cs = T[i+j'-1]$ denote the letter itself. The shift is then
\begin{equation}
\label{eq:bc-indet-good}
\textsc{Indet-BCS-Fast(BM)}=
\begin{cases}
\max\bigl(1,\, j-\tau(c)\bigr),
  & \text{if } T[i+j-1]=c \text{ is solid},\\[4pt]
\max\bigl(1,\, j'-\tau(cs)\bigr),
  & \text{if } T[i+j-1]=\lambda \text{ is indeterminate}\\
  & \quad \text{and such a } j' \text{ exists},\\[4pt]
1,
  & \text{otherwise.}
\end{cases}
\end{equation}

The bad character rules given so far reuse the standard bad-character table $\tau$,
which stores only the rightmost occurrence of each character and is thus
insensitive to \emph{where} in the pattern a mismatch occurs. We now
develop a family of rules built on a richer, position-indexed table
$\tau'$ that, for each character, records its rightmost occurrence to the
left of a given position. This additional information supports stronger,
position-aware shifts and lets us handle an indeterminate mismatch
without collapsing it to a single solid letter, at the cost of
$O(m|\Sigma|)$ space in place of $O(|\Sigma|)$. Formally, the extended
bad-character table (position-indexed)~\cite{Gusfield1997} for indeterminate
patterns, $\tau' : \Sigma \times \{1,\ldots,m\} \rightarrow
\{0,1,\ldots,m-1\}$, is defined as

\begin{equation}\label{eq:bm-bc-tau-ext}
\tau'(c, k) =
\begin{cases}
\max \{\, i \mid 1 \le i < k \text{ and } c \in P[i] \,\}, & \text{if such } i \text{ exists}, \\[6pt]
0, & \text{otherwise.}
\end{cases}
\end{equation}
Here $\tau'(c,k)$ is the rightmost occurrence of a character $c \in
\Sigma$ in the prefix $P[1..k-1]$, for every $1 \le k \le m$ (and $0$
when $c$ does not occur in that prefix). In effect, $\tau'$ extends the
position-indexed bad-character table of Tarhio and Ukkonen~\cite{TU93} to indeterminate patterns by replacing equality with membership; Dehghani et al.~\cite{Neerja2024} use this table to compute the bad-character shift for their \textsc{BM\_Indet} algorithm, as follows.

If a mismatch occurs at position $j$ of the pattern against the text character $T[i+j-1] = \lambda$, the bad-character (BC) shift is
\begin{equation}\label{eq:bm-Indet-shift}
\textsc{Indet-BCS(BM\_Indet)} \;=\; j - \max\{\, \tau'(c,j) : c \in \lambda \,\}.
\end{equation}
The inner maximum equals the rightmost position $j' < j$ with
$P[j'] \cap \lambda \neq \emptyset$, so
$j - j'$ is the smallest shift that realigns $\lambda$ with a pattern letter it can
match; every skipped alignment is a guaranteed mismatch. Moreover, since $\tau'(c,j) \le j-1$,
the shift is always at least $1$. 

We use this extended position-indexed bad character table to compute the shifts for our three new rules presented below.

\noindent\textbf{BC Rule II: Generalization of \textsc{BM\_Indet}.}
Rule~II simply combines the two shifts we already have the fast solid-letter
rule \textsc{Indet-BCS-Fast(BM)} of Rule~I (Equation~\ref{eq:bc-indet-good})
and the bad-character rule \textsc{Indet-BCS(BM\_Indet)} of
Equation~\ref{eq:bm-Indet-shift} and applies the larger of the two:

\begin{equation}
\label{eq:bc-rule1-envelope}
\begin{aligned}
\text{\textsc{Indet-BCS(GBM\_Indet)}}
&= \max\bigl(\text{\textsc{Indet-BCS(BM\_Indet)}},\\
&\qquad\qquad \text{\textsc{Indet-BCS-Fast(BM)}}\bigr).
\end{aligned}
\end{equation}

\noindent\textbf{BC Rule III: Edge Solid Characters.}
Let $S = \{\, p : j < p \le m \text{ and } T[i+p-1] \text{ is solid} \,\}$, with
first and last solid positions $j_1 = \min S$ and $j_2 = \max S$. For a boundary position $p \in \{j_1, j_2\}$ with solid letter $c_p = T[i+p-1]$, write
$\delta(p) = p - \tau'(c_p, p)$ when $\tau'(c_p, p) > 0$, and $\delta(p) = p$
otherwise; in either case $\delta(p) \ge 1$ realigns $c_p$ with its rightmost
earlier occurrence or moves it past the pattern. The resulting shift is,
\begin{equation}
\label{eq:bc-edge}
\textsc{Indet-BCS(Rule~III)} =
\begin{cases}
\max\bigl(\delta(j_1),\, \delta(j_2)\bigr), & |S| \ge 2, \\[6pt]
\delta(j_1),                               & |S| = 1, \\[6pt]
1,                                          & S = \emptyset.
\end{cases}
\end{equation}

\noindent\textbf{BC Rule IV: Least-Frequent Solid Character.}
Rule~IV refines the edge idea of Rule~III. Rather than using the solid
letters at the two ends of the matched suffix as the lookup characters, it uses
the \emph{least frequently occurring} solid letter within the matched suffix,
whose sparse occurrences in $P$ tend to yield a larger shift. Let
$\phi(c) = |\{\, k : 1 \le k \le m \text{ and } c \in P[k] \,\}|$ be the
frequency of a letter $c$ in $P$. This can be computed by a simple left to right scan in $O(m\sigma)$-time. Let
$S = \{\, k : j < k \le m \text{ and } T[i+k-1] \text{ is solid} \,\}$ be the
solid positions of the matched suffix of $T$, with solid letters $C=\{T[i+k-1]: k \in S\}$. Let $f(c)=\min\{k \in S \mid T[i+k-1] = c\}$ be the position of its first occurrence in the matched suffix of $T$. We select a least-frequent solid
letter; if more than one attains the minimum frequency, we choose the letter $c$
with the largest $f(c)$, and set $j' = f(c)$. The resulting shift is

\begin{equation}
\label{eq:bc-freq}
\textsc{Indet-BCS(Rule~IV)} = 
\begin{cases}
\max\bigl(1,  j' - \tau'(c, j')\bigr), & \text{ if } S \neq \emptyset, \\[6pt]
1,                   & \text{ otherwise}.
\end{cases}
\end{equation}

We now present the function \textsc{Fast\_\allowbreak GSR\_\allowbreak Indet\_\allowbreak Shift} shown in Algorithm~\ref{alg:fast-bmindet-shift}, which computes the good suffix rule when either the matched suffix in the text or the matched suffix in the pattern is indeterminate. The \textsc{Fast\_\allowbreak GSR\_\allowbreak Indet\_\allowbreak Shift} on a mismatch computes
the good-suffix shift from the two \emph{solid edge letters} of the matched
suffix, the first and last solid letters $c_1$ and $c_2$ at pattern positions
$j_1\le j_2$ with distance $d=j_2-j_1$, together with the position-indexed table
$\tau'$. Because the pattern shifts rigidly, the matched suffix can recur only
where $c_1$ and $c_2$ recur in $P$ at the same distance $d$. The algorithm
therefore walks the earlier occurrences of $c_1$ in decreasing order through
$\tau'$ and, for each occurrence $x_1$, tests whether $c_2$ occurs at $x_1+d$; the
first occurrence that passes gives the shift $j_1-x_1$, which is the smallest that
realigns both solid letters. If $c_1$ has no earlier occurrence the pattern is shifted
past by $j_1$; a single solid letter is realigned directly through $\tau'$, and a
matched suffix with no solid letter shifts by $1$. Unlike \textsc{Indet\_GSR\_Shift} of \textsc{BM\_Indet}~\cite{Neerja2024}, which
rebuilds a prefix array at every indeterminate alignment, the shift here is read
from a table in constant time and evaluated, so no prefix array is constructed during the search.

\begin{algorithm}[t!]
\caption{\textsc{Fast\_GSR\_Indet\_Shift}$(\tau',j_1,j_2)$}
\label{alg:fast-bmindet-shift}
\begin{algorithmic}[1]
\Function{Fast\_GSR\_Indet\_Shift}{$\tau',j_1,j_2$}: Integer
 \Comment{Let $c_1, c_2$ be the solid edge letters in $T[i{+}j \dots i{+}m{-}1]$
 at pattern positions $j_1, j_2 \geq j_1 \in \{j{+}1, \dots, m\}$}
    \If{$j_1 =0 \wedge j_2 = 0$}\Comment{no solid letters in matched suffix in $T$}
        \State \Return $1$ \Comment{no solid letters}
    \EndIf
    \State $c_1 \gets T[i{+}j_1{-}1]$
    \State $c_2 \gets T[i{+}j_2{-}1]$
    \State $x_1 \gets \tau'(c_1, j_1)$ \Comment{right most occurrence of $c_1$ in $P[1..j_1-1]$}
    \If{$j_1 = j_2$} \Comment{only one solid letter in matched suffix in $T$}
        
        \State \Return $\mathit{gs\_shift} \gets (j_1 - x_1)$ 
    \Else
    \State $d \gets j_2 - j_1$ 
    \State $x_2 \gets \tau'(c_2, x_1+d+1)$ \Comment{right most occurrence of $c_2$ in $P[1..x_1+d]$}
    \While{$x_2 \neq x_1+d \wedge x_1 \neq 0 $}
    \State $x_1 \gets \tau'(c_1,x_1)$
    \State $x_2 \gets \tau'(c_2, x_1+d+1)$
    \EndWhile
       \State \Return $\mathit{gs\_shift} \gets (j_1 - x_1)$
    \EndIf
\EndFunction
\end{algorithmic}
\end{algorithm}

We devise four fast algorithms, which we call the
\textit{\textsc{Fast\_BM\_Indet} algorithm variants}, by pairing each of
BC~Rules~I--IV with the new good-suffix rule. Specifically, a
\textsc{Fast\_BM\_Indet} variant is obtained from the \textsc{BM\_Indet}
algorithm of~\cite{Neerja2024} by replacing its
\textsc{bad\_character\_rule\_shift} with one of BC~Rules~I--IV and its
\textsc{indet\_gsr\_shift} with \textsc{Fast\_\allowbreak GSR\_\allowbreak
Indet\_\allowbreak Shift}, as shown in Algorithm~\ref{alg:fastbmindet} in the
Appendix.

\begin{theorem}\label{thm:fastbm-correct}All the \textsc{Fast\_BM\_Indet} algorithm variants correctly identify all occurrences of an indeterminate pattern $P$ in an indeterminate text $T$.
\end{theorem}

\begin{proof}
A shift $s$ at alignment $i$ is \emph{safe} if $P$ does not occurs at any alignment $i+t$, where $1\le t<s$. Under shift $t$, a solid text letter $c$ aligned with a pattern letter that does not contain $c$, precludes an occurrence of $P$ at $a+t$. The bad-character shift advances $P$ only past alignments placing its inspected text characters against cells lacking them, and is therefore safe. When $P$ and the matched suffix are solid, matching reduces to equality and the classical good-suffix shift is safe. Otherwise, Algorithm~\ref{alg:fast-bmindet-shift} inspects earlier occurrences of the first solid anchor $c_1$ in order of increasing shift and returns the first where the second anchor $c_2$ appears at distance $d$; every smaller shift it skips places $c_1$ or $c_2$ against a letter that does not contain it and is thus safe, while returning $j_1$ (no such occurrence) or $1$ (no solid letter) without skipping a matching position. Because the applied shift is the maximum of safe shifts, it is safe, and each reported position is confirmed by verifying all $m$ characters. Thus, no occurrence of $P$ is skipped or falsely reported.
\qed
\end{proof}

\begin{theorem}\label{thm:fastbm}
The \textsc{Fast\_BM\_Indet} algorithm variants require at most $O(nm+ m \sigma^2)$ time and $O(m \sigma)$ space.
\end{theorem}
\begin{proof}
The function \textsc{Fast\_\allowbreak GSR\_\allowbreak Indet\_\allowbreak Shift} computes the good-suffix shift in $O(m)$ time per alignment.
The standard bad-character table $\tau$ can be built in $O(m+\sigma)$ time (Lemma~\ref{lem:bm-complexity}). The position-indexed table $\tau'$ has $m\sigma$ entries; since testing membership $c\in P[i]$ for an indeterminate letter costs $O(\sigma)$ time, constructing $\tau'$ requires $O(m\sigma^2)$ time and $O(m\sigma)$ space.
At each alignment, the bad-character shift is computed in $O(1)$ time for BC~Rule~I, in $O(\sigma)$ time for BC~Rule~II (taking a maximum over the letters in an indeterminate text cell), and in $O(1)$ time for BC~Rules~III--IV.
Since there are at most $n$ alignments, the total running time for these variants range from $O(n(m+\sigma) +m\sigma^2)$ to $O(nm +m\sigma^2)$ in the worst case.\qed

\end{proof}
Thus \textsc{Fast\_BM\_Indet} improves on \textsc{BM\_Indet}~\cite{Neerja2024} by a factor of $O(\sqrt m)$ in the worst case and is faster across all tested inputs.

\section{Experimental Evaluation}
\label{sec:experimental-results}
In this section, we evaluate two families of algorithms. First, we compare seven
classical bad-character variants (Horspool, Sunday, Raita, Smith, Tuned BM,
Zhu--Takaoka, Berry--Ravindran), the algorithm of Holub et al.~\cite{HolubSmythWang2008},
and brute force. Second, we evaluate \textsc{Fast\_BM\_Indet} combined with each of
our bad-character Rules~I--IV and with Zhu--Takaoka, and compare these hybrids with
\textsc{BM\_Indet} and \textsc{KMP\_Indet}~\cite{Neerja2024}, and with brute force. We omit
DBWT~\cite{Daykin2019}, which was shown in~\cite{Neerja2024} to be substantially slower
than \textsc{BM\_Indet}. Experimental data consist of random strings over alphabets
$\sigma\in\{4,9,20\}$ and the \textit{E.\,coli} K-12 MG1655 genome (NCBI RefSeq
NC\_000913.3) with natural IUPAC ambiguity codes, all encoded using the power-of-two
method of Section~\ref{sec:preliminaries}. All algorithms are implemented in C++
(GCC~6.3.0, \texttt{-O2}) and run single-threaded on an Intel Core i7-1165G7
(Windows~11). Following Dehghani et al.~\cite{Neerja2024}, text length $n=1000i$
and pattern length $m=40i$ grow with a parameter $i$, with $0.06i$ and $4i$
indeterminate letters (rounded). Each point is the mean over $100$ runs (ten
inputs, each run ten times). We report results for short texts ($i=1,\ldots,10$)
and long texts ($i=100,200,\ldots,1000$, up to $n=10^6$), plotting running time in
milliseconds against $n$ for all algorithms on identical inputs.

\begin{figure}[t]
\centering
\includegraphics[width=\linewidth]{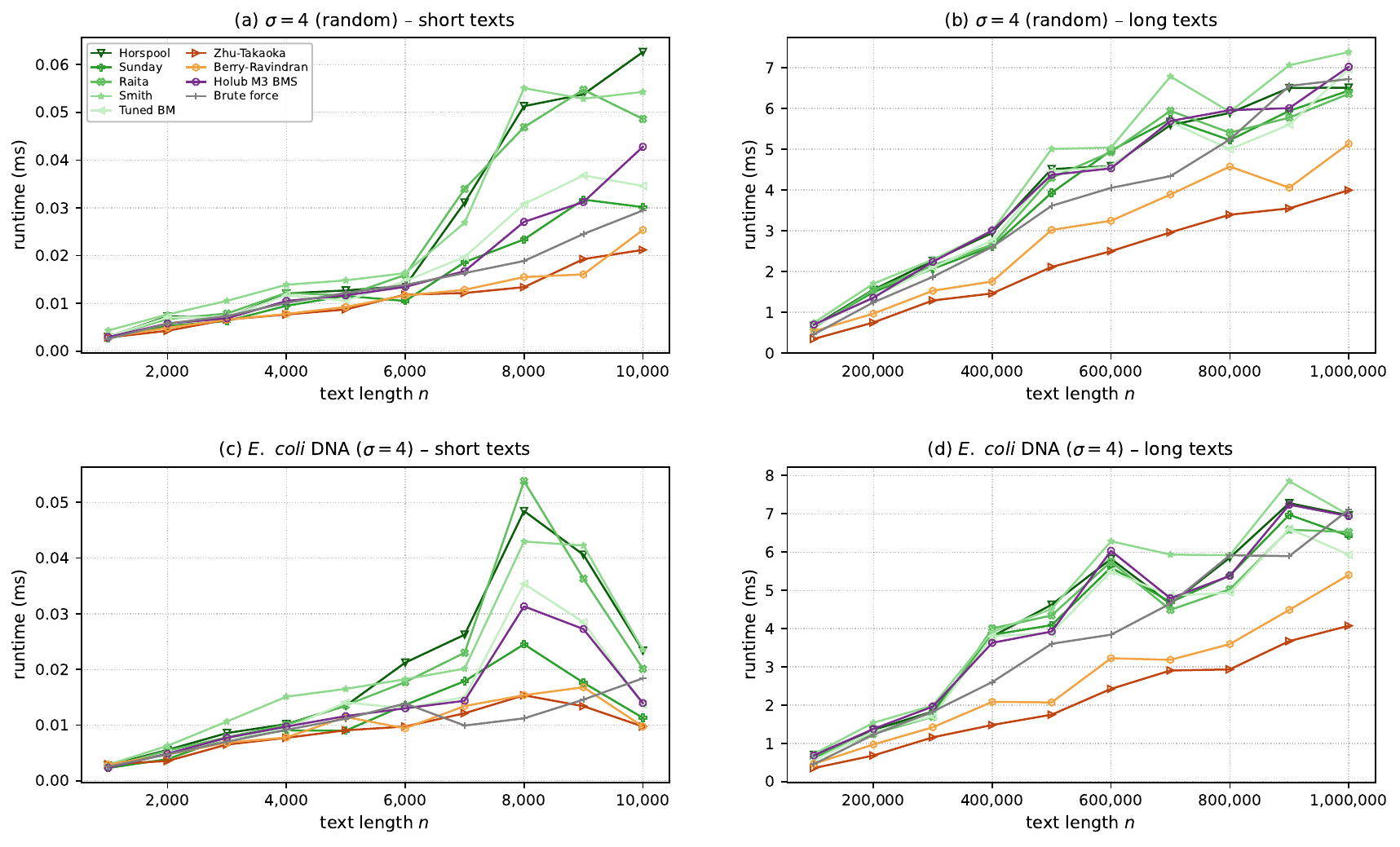}
\caption{Running time (ms) against text length $n$ of bad character only algorithms, which include the bad-character-only adaptations of Horspool, Sunday, Raita, Smith, Tuned BM, Zhu--Takaoka and Berry--Ravindran, the algorithm of Holub et al.\ (Holub M3 BMS), and brute force. Inputs are random strings with $\sigma=4$ (a, b) and the \textit{E.\,coli} genome, $\sigma=4$ (c, d); short texts on the left and long texts on the right.}
\label{fig:bc-1}
\end{figure}

\begin{figure}[t]
\centering
\includegraphics[width=\linewidth]{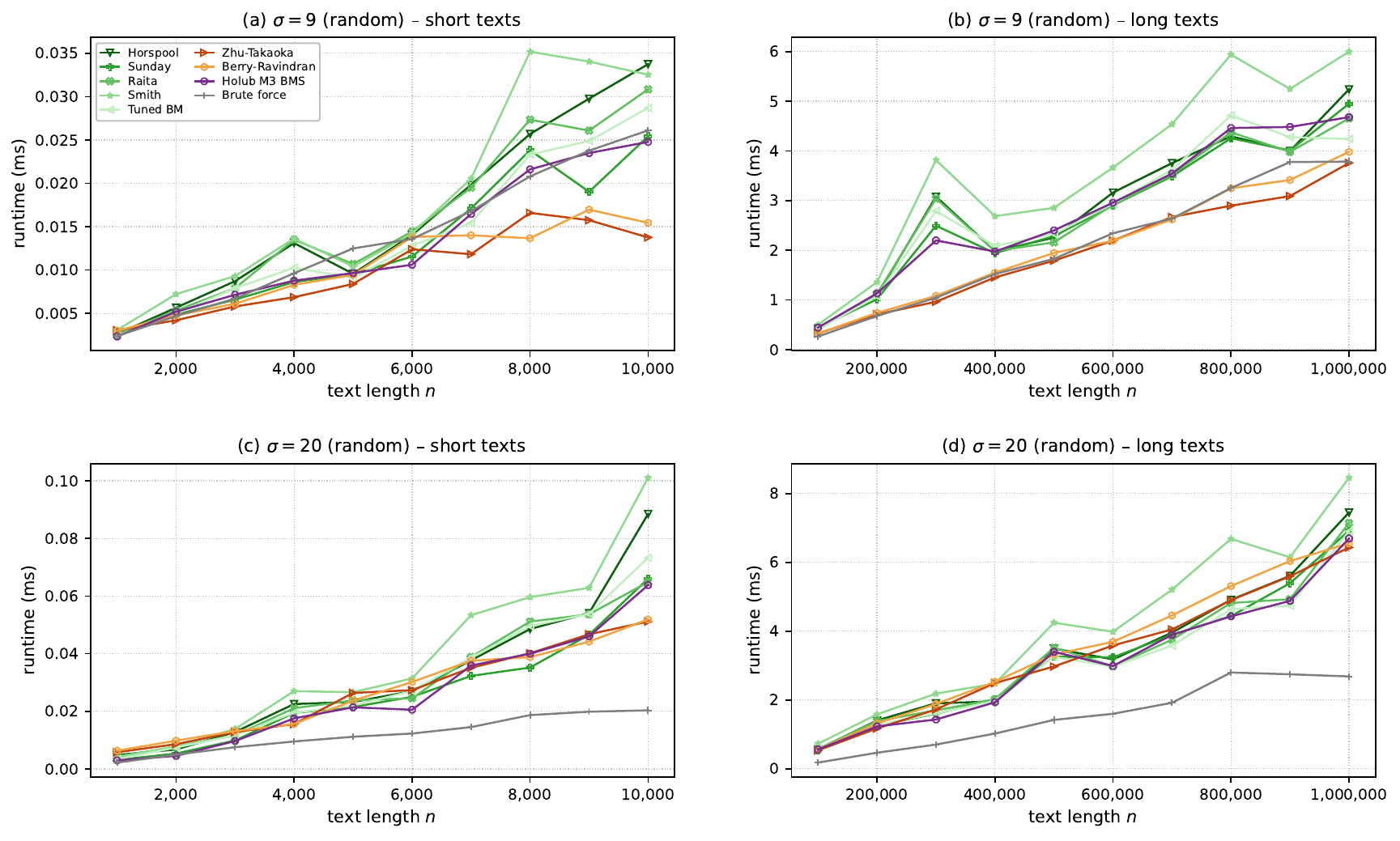}
\caption{Running time (ms) against text length $n$ of all algorithms referred to in Figure~\ref{fig:bc-1}, on random strings with $\sigma=9$ (a, b) and $\sigma=20$ (c, d); short texts on the left and long texts on the right.}
\label{fig:bc-2}
\end{figure}
\noindent\textbf{Scaling of the Bad-Character Rules:} We now compare the bad-character-only algorithms, with brute force as a reference, in Figures~\ref{fig:bc-1} and~\ref{fig:bc-2}. The character-pair variants, Zhu--Takaoka and Berry--Ravindran, are the fastest on the small alphabets ($\sigma=4$ and the \textit{E.\,coli} DNA): a letter pair occurs in the pattern less often than a single letter, so it yields larger shifts. The advantage of the pair shrinks as the alphabet grows, and at $\sigma=20$ all Boyer--Moore variants lie in a narrow band. Brute force runs in near-linear time on random texts, since most alignments fail at the first comparison, so a shift rule helps only when the alignments it skips outweigh the cost of computing it. This holds for the pair variants on the small alphabets, whereas the single-character variants are no better than brute force. At $\sigma=9$ the pair variants only match it, and at $\sigma=20$ it is the fastest algorithm, because shift computation inspects every letter of an indeterminate letter requiring additional $O(\sigma)$ time. Bad-character rules alone thus pay off only on small alphabets, the setting of genomic data, which motivates their combination with the good-suffix rule. Among the pair variants, Zhu--Takaoka is the more dependable: it is the lowest curve at every long-text size on the small alphabets, and elsewhere it is ahead of or level with Berry--Ravindran. We therefore use Zhu--Takaoka as the bad-character rule in the remaining experiments.

\begin{figure}[t]
\centering
\includegraphics[width=\linewidth]{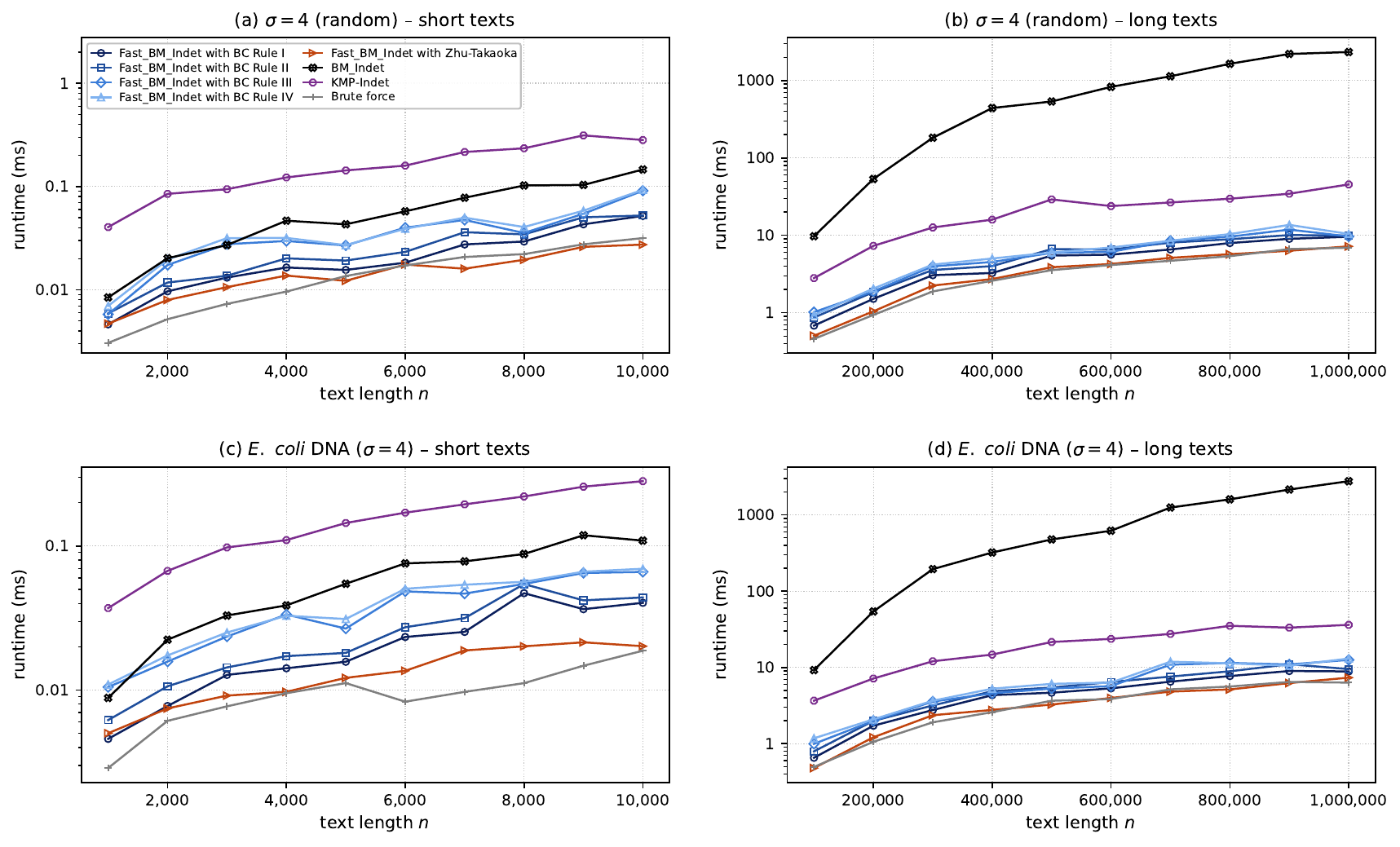}
\caption{Running time (logarithmic scale) of all \textsc{Fast\_BM\_Indet} variants and the good-suffix rule with the adapted Zhu--Takaoka variant, against \textsc{BM\_Indet}~\cite{Neerja2024}, \textsc{KMP\_Indet}~\cite{Neerja2024} and brute force, on random strings with $\sigma=4$ ($a, b$) and on the \textit{E.\,coli} genome ($c, d$), for short texts (left) and long texts (right).}
\label{fig:gs-1}
\end{figure}

\begin{figure}[t]
\centering
\includegraphics[width=\linewidth]{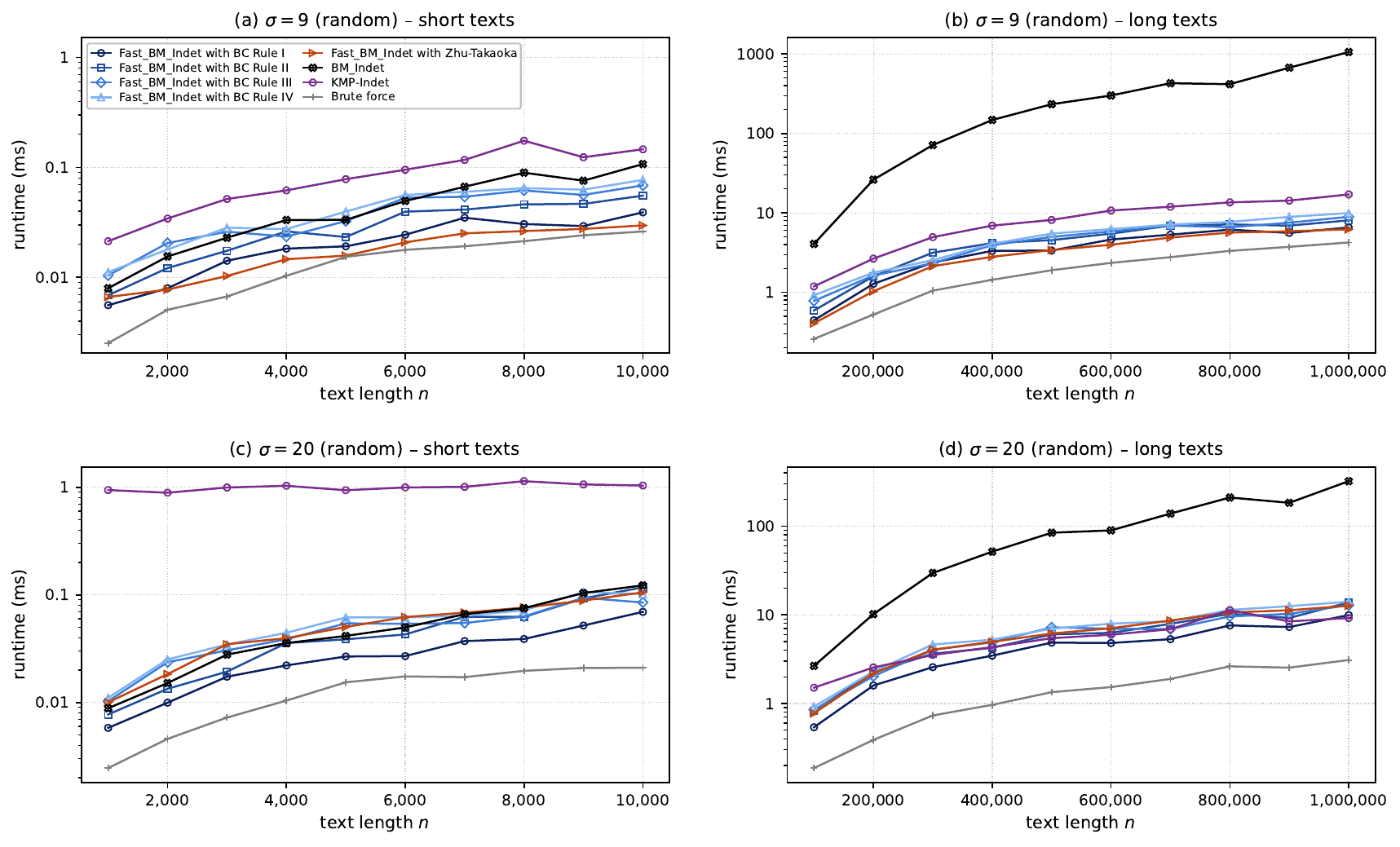}
\caption{Running time (logarithmic scale) of all \textsc{Fast\_BM\_Indet} variants and the good-suffix rule with the adapted Zhu--Takaoka variant, against \textsc{BM\_Indet}~\cite{Neerja2024}, \textsc{KMP\_Indet}~\cite{Neerja2024} and brute force on random strings with $\sigma=9$ ($a, b$) and $\sigma=20$ ($c, d$), for short texts (left) and long texts (right).}
\label{fig:gs-2}
\end{figure}

\noindent\textbf{Combining Bad-Character Rules with the Good-Suffix Rule:}
We evaluate our four \textsc{Fast\_BM\_Indet} variants (BC~Rules~I--IV paired with
the new good-suffix rule), together with a fifth variant combining the good-suffix
rule with the adapted Zhu--Takaoka bad-character rule, the best-performing BC-only. We benchmark all five against
\textsc{BM\_Indet}~\cite{Neerja2024}, \textsc{KMP\_Indet}~\cite{KMPIndet}, and brute
force. Figures~\ref{fig:gs-1} and~\ref{fig:gs-2} (log scale) show that every
\textsc{Fast\_BM\_Indet} variant is dramatically faster than \textsc{BM\_Indet}; on
long texts the gap widens with $n$, reaching two orders of magnitude (e.g., at
$n=10^6$ on $\sigma=4$, $2{,}333$\,ms for \textsc{BM\_Indet} versus $7$\,ms for the
Zhu--Takaoka hybrid). The hybrids also beat \textsc{KMP\_Indet} on small alphabets
and short texts, by about $5\times$ on DNA at $n=10^6$; only at $\sigma=20$ does
\textsc{KMP\_Indet} catch up on long texts. Among Rules~I--IV, Rule~I is clearly
best; Rules~II--IV are slower by $20$--$90\%$ because they scan the matched suffix
or every letter of an indeterminate cell to compute shifts that are rarely larger.
Rule~I is within $10\%$ of Zhu--Takaoka on small alphabets and is the fastest
hybrid at $\sigma=20$. It uses a single $O(\sigma)$ table and one $O(1)$ lookup per
alignment, whereas Zhu--Takaoka requires an $O(\sigma^2)$ pair table. Thus Rule~I
is preferable on large alphabets and Zhu--Takaoka on small ones, where the pair
table is cheap. Because most alignments fail immediately on random strings and on
\textit{E.\,coli}, brute force is a strong baseline: Zhu--Takaoka with
\textsc{Fast\_GSR\_Indet\_Shift} matches it on $\sigma=4$ and DNA, and is slower by
factors of $1.5$ at $\sigma=9$ and $4$ at $\sigma=20$. Our hybrids therefore bring
indeterminate Boyer--Moore matching to the speed of brute force on such inputs,
whereas \textsc{BM\_Indet} remains far slower, though they do not surpass brute
force there.

\bibliographystyle{splncs04}
\bibliography{sample}

@book{Gusfield1997,
    author    = {Gusfield, Dan},
    title     = {Algorithms on Strings, Trees, and Sequences: Computer Science
                 and Computational Biology},
    publisher = {Cambridge University Press},
    address   = {Cambridge, UK},
    year      = {1997}
}

@book{Navarro2002,
    author    = {Navarro, Gonzalo and Raffinot, Mathieu},
    title     = {Flexible Pattern Matching in Strings: Practical On-Line Search
                 Algorithms for Texts and Biological Sequences},
    publisher = {Cambridge University Press},
    address   = {Cambridge, UK},
    year      = {2002}
}

@article{BaezaYates1992,
    author  = {Baeza-Yates, Ricardo and Gonnet, Gaston H.},
    title   = {A new approach to text searching},
    journal = {Communications of the ACM},
    volume  = {35},
    number  = {10},
    pages   = {74--82},
    year    = {1992}
}

@incollection{FischerPaterson1974,
    author    = {Fischer, Michael J. and Paterson, Michael S.},
    title     = {String matching and other products},
    booktitle = {Complexity of Computation},
    publisher = {American Mathematical Society},
    pages     = {113--125},
    year      = {1974}
}

@article{Abrahamson1987,
    author  = {Abrahamson, Karl},
    title   = {Generalized string matching},
    journal = {SIAM Journal on Computing},
    volume  = {16},
    number  = {6},
    pages   = {1039--1051},
    year    = {1987}
}

@inproceedings{Iliopoulos2016,
    author    = {Iliopoulos, Costas S. and Radoszewski, Jakub},
    title     = {Truly subquadratic-time extension queries and periodicity
                 detection in strings with uncertainties},
    booktitle = {Proceedings of the 27th Annual Symposium on Combinatorial
                 Pattern Matching (CPM)},
    series    = {LIPIcs},
    volume    = {54},
    pages     = {8:1--8:12},
    year      = {2016}
}

@article{FastSearch2005,
    author  = {Cantone, Domenico and Faro, Simone},
    title   = {Fast-Search: A new efficient variant of the {Boyer-Moore}
               string search algorithm},
    journal = {Proceedings in Informatics},
    volume  = {14},
    pages   = {47--58},
    year    = {2003}
}

@inproceedings{WuManber1994,
    author    = {Wu, Sun and Manber, Udi},
    title     = {A fast algorithm for multi-pattern searching},
    institution = {Department of Computer Science, University of Arizona},
    number    = {TR-94-17},
    year      = {1994}
}

@article{ReverseFactor1994,
    author  = {Crochemore, Maxime and Czumaj, Artur and Gasieniec, Leszek
               and Lecroq, Thierry and Plandowski, Wojciech and Rytter,
               Wojciech},
    title   = {Sublinear-time string matching and related problems},
    journal = {SIAM Journal on Computing},
    volume  = {22},
    number  = {4},
    pages   = {792--807},
    year    = {1993}
}

@inproceedings{TurboReverseFactor1994,
    author    = {Lecroq, Thierry},
    title     = {Experimental results on string matching algorithms},
    journal   = {Software: Practice and Experience},
    volume    = {25},
    number    = {7},
    pages     = {727--765},
    year      = {1995}
}

@article{Knuth1977,
    author  = {Knuth, Donald E. and Morris, James H. and Pratt, Vaughan R.},
    title   = {Fast pattern matching in strings},
    journal = {SIAM Journal on Computing},
    volume  = {6},
    number  = {2},
    pages   = {323--350},
    year    = {1977}
}

@article{Boyer1977,
    author  = {Boyer, Robert S. and Moore, J. Strother},
    title   = {A fast string searching algorithm},
    journal = {Communications of the ACM},
    volume  = {20},
    number  = {10},
    pages   = {762--772},
    year    = {1977}
}

@article{TU93,
  author  = {Tarhio, Jorma and Ukkonen, Esko},
  title   = {Approximate {Boyer--Moore} String Matching},
  journal = {SIAM Journal on Computing},
  volume  = {22},
  number  = {2},
  pages   = {243--260},
  year    = {1993},
  
}

@article{btt2013,
  author    = {Baker, Brenda S. and Thierry, Emilie and Tischler, Germ{\'a}n},
  title     = {Pattern matching with indeterminate strings: Algorithms and data structures},
  journal   = {Theoretical Computer Science},
  volume    = {525},
  pages     = {28--42},
  year      = {2014},
 
}

@article{Horspool1980,
    author  = {Horspool, R. Nigel},
    title   = {Practical fast searching in strings},
    journal = {Software: Practice and Experience},
    volume  = {10},
    number  = {6},
    pages   = {501--506},
    year    = {1980}
}

@article{Sunday1990,
    author  = {Sunday, Daniel M.},
    title   = {A very fast substring search algorithm},
    journal = {Communications of the ACM},
    volume  = {33},
    number  = {8},
    pages   = {132--142},
    year    = {1990}
}

@article{Raita1992,
    author  = {Raita, Timo},
    title   = {Tuning the {Boyer-Moore-Horspool} string searching algorithm},
    journal = {Software: Practice and Experience},
    volume  = {22},
    number  = {10},
    pages   = {879--884},
    year    = {1992}
}

@article{Smith1991,
    author  = {Smith, Peter D.},
    title   = {Experiments with a very fast substring search algorithm},
    journal = {Software: Practice and Experience},
    volume  = {21},
    number  = {10},
    pages   = {1065--1074},
    year    = {1991}
}

@article{Hume1991,
    author  = {Hume, Andrew and Sunday, Daniel},
    title   = {Fast string searching},
    journal = {Software: Practice and Experience},
    volume  = {21},
    number  = {11},
    pages   = {1221--1248},
    year    = {1991}
}

@article{ZhuTakaoka1987,
    author  = {Zhu, Rui Feng and Takaoka, Tadao},
    title   = {On improving the average case of the {Boyer-Moore} string
               matching algorithm},
    journal = {Journal of Information Processing},
    volume  = {10},
    number  = {3},
    pages   = {173--177},
    year    = {1987}
}

@inproceedings{BerryRavindran1999,
    author    = {Berry, Thomas and Ravindran, Somsubhra},
    title     = {A fast string matching algorithm and experimental results},
    booktitle = {Prague Stringology Club Workshop},
    pages     = {16--26},
    year      = {1999}
}

@article{Crochemore1992,
    author  = {Crochemore, Maxime and Czumaj, Artur and Gasieniec, Leszek
               and Jarominek, Stefan and Lecroq, Thierry and Plandowski,
               Wojciech and Rytter, Wojciech},
    title   = {Speeding up two string-matching algorithms},
    journal = {Algorithmica},
    volume  = {12},
    number  = {4--5},
    pages   = {247--267},
    year    = {1992}
}

@article{HolubSmythWang2008,
    author  = {Holub, Jan and Smyth, W. F. and Wang, Shu},
    title   = {Fast pattern-matching on indeterminate strings},
    journal = {Journal of Discrete Algorithms},
    volume  = {6},
    number  = {1},
    pages   = {37--50},
    year    = {2008}
}

@article{Crochemore2016,
    author  = {Crochemore, Maxime and Iliopoulos, Costas S. and Kundu, Ritu
               and Mohamed, Manal and Vayani, Fatima},
    title   = {Linear algorithm for conservative degenerate pattern matching},
    journal = {Engineering Applications of Artificial Intelligence},
    volume  = {51},
    pages   = {109--114},
    year    = {2016}
}

@article{Daykin2019,
    author  = {Daykin, Jacqueline W. and Groult, Richard and Guesnet, Yannick
               and Lecroq, Thierry and Lefebvre, Arnaud and L{\'e}onard, Martine
               and Mouchard, Laurent and Prieur-Gaston, {\'E}lise
               and Watson, Bruce},
    title   = {Efficient pattern matching in degenerate strings with the
               {Burrows-Wheeler} transform},
    journal = {Information Processing Letters},
    volume  = {147},
    pages   = {82--87},
    year    = {2019}
}

@article{Alzamel2020,
    author  = {Alzamel, Mai and Ayad, Lorraine A. K. and Bernardini, Giulia
               and Grossi, Roberto and Iliopoulos, Costas S. and Pisanti, Nadia
               and Pissis, Solon P. and Rosone, Giovanna},
    title   = {Comparing degenerate strings},
    journal = {Fundamenta Informaticae},
    volume  = {175},
    number  = {1--4},
    pages   = {41--58},
    year    = {2020}
}

@inproceedings{KMPIndet,
    author    = {Mhaskar, Neerja and Smyth, W. F.},
    title     = {Simple {KMP} pattern-matching on indeterminate strings},
    booktitle = {Proceedings of the Prague Stringology Conference (PSC)},
    pages     = {125--133},
    year      = {2020}
}

@article{Neerja2024,
    author  = {Dehghani, Hossein and Lecroq, Thierry and Mhaskar, Neerja
               and Smyth, W. F.},
    title   = {Practical {KMP/BM} style pattern-matching on indeterminate strings},
    journal = {Theoretical Computer Science},
    volume  = {1027},
    pages   = {64--79},
    year    = {2025},

}

@inproceedings{Gawrychowski2020,
    author    = {Gawrychowski, Pawe{\l} and Ghazawi, Samah and Landau, Gad M.},
    title     = {On indeterminate strings matching},
    booktitle = {Proceedings of the 31st Annual Symposium on Combinatorial
                 Pattern Matching (CPM 2020)},
    series    = {LIPIcs},
    volume    = {161},
    pages     = {14:1--14:14},
    year      = {2020},
   
}

@article{Gawrychowski2026,
    author  = {Gawrychowski, Pawe{\l} and Ghazawi, Samah and Landau, Gad M.},
    title   = {On the complexity of indeterminate strings matching},
    journal = {Theoretical Computer Science},
    volume  = {1067},
    pages   = {115771},
    year    = {2026},
    
}

@INPROCEEDINGS{HMZ2026, 
 author = {Hamed Hasibi and Neerja Mhaskar and Tieyun Zhang},
 title = {Faster Indeterminate Pattern Matching Algorithm with Practical Implementations},
 booktitle = {Proceedings of the Prague Stringology Conference 2026},
 year = {2026}
}

@article{IUPAC1985,
    author  = {{IUPAC-IUB Commission on Biochemical Nomenclature}},
    title   = {Nomenclature for incompletely specified bases in nucleic acid
               sequences},
    journal = {European Journal of Biochemistry},
    volume  = {150},
    pages   = {1--5},
    year    = {1985}
}

@article{Stormo2000,
    author  = {Stormo, Gary D.},
    title   = {{DNA} binding sites: representation and discovery},
    journal = {Bioinformatics},
    volume  = {16},
    number  = {1},
    pages   = {16--23},
    year    = {2000}
}

\clearpage
\appendix
\section{Appendix}
\label{Appendix}
\subsection{Boyer Algorithm}
For completeness, we provide the bad character rule and good suffix rule computation details of the  Boyer--Moore algorithm~\cite{Boyer1977} below:

\itbf{Bad Character (BC) Rule.}  To compute the shift for this rule, a bad-character table $\tau : \Sigma \rightarrow \{0,1,\ldots,m\}$ is
defined as
\begin{equation}\label{eq:bm-bc-tau}
\tau(c) =
\begin{cases}
\max \{\, i \mid 1 \le i \le m \text{ and } p[i] = c \,\}, & \text{if } c \in p, \\[6pt]
0, & \text{if } c \notin p.
\end{cases}
\end{equation}

 If a mismatch occurs at position $j$ of the pattern against text character
$c$, then $c$ is termed as the \textit{lookup character}, and the bad-character (BC) shift
  is given by $\max(1,\, j - \tau[c])$.

\itbf{Good Suffix (GS) Rule.}
Suppose $p[j+1..m]$ is the longest matched suffix of length $m-j$ at an alignment $i$. We define two tables $L$ and $\ell$ w.r.t $p$ as:
\[
L[j+1] =
\begin{cases}
\max\{\, k \mid 1 \leq k \leq j
\text{ and } p[k..m-j+k-1] = p[j+1..m] \,\}, & \text{if such } k
\text{ exists,} \\[6pt]
0, & \text{otherwise.}
\end{cases}
\]

At an index position $j+1$, $L[j+1]$ stores the starting position $k$ of the rightmost occurrence of the longest matched suffix $p[j+1..m]$ in the unmatched prefix $p[1..j]$ of the pattern. The other table, $\ell$, is defined as follows:
\[
\ell[j+1] =
\begin{cases}
\max\{\, r \mid 1 \leq r \leq m-j
\text{ and } p[1..r] = p[m-r+1..m] \,\}, & \text{if such } r
\text{ exists,} \\[6pt]
0, & \text{otherwise.}
\end{cases}
\]
The table $\ell$, at index position $j+1$, stores the length of the longest prefix of $p$ that is also a suffix of $p[j+1..m]$.
The Good Suffix (GS) shift is then computed by,

\begin{equation}\label{eq:bm-gs-shift1}
\text{GS shift} =
\begin{cases}
j+1 - L[j+1], & \text{if } L[j+1] > 0, \\[6pt]
m - \ell[j+1], & \text{otherwise.}
\end{cases}
\end{equation} 

\subsection{Summary of classical Boyer--Moore Variants}
\begin{table}[h!]
\centering
\caption{Taxonomy of classical Boyer--Moore variants.}
\label{tab:taxonomy}
\footnotesize
\setlength{\tabcolsep}{3pt}
\resizebox{\linewidth}{!}{%
\begin{tabular}{p{1.2cm}p{2.3cm}p{4.5cm}p{5cm}}
\toprule
 & \textbf{BM Variant}
 & \textbf{Look up character in $t$}
 & \textbf{When it helps?} \\
\midrule
\multirow{5}{=}{\raggedright BC-Only\\single-char}
 & Horspool~\cite{Horspool1980}
 & Last text char in window -- $t[i{+}m{-}1]$
 & General purpose; simple baseline; widely used \\[4pt]
 & Sunday~\cite{Sunday1990}
 & First text char past window -- $t[i{+}m]$
 & Large alphabets; often yields larger shifts than Horspool \\[4pt]
 & Raita~\cite{Raita1992}
 & Last, first and middle text chars of window $t[i+m-1]$, $t[i]$, and $t[ i+\lfloor\frac{(m)}{2}\rfloor]$
 & Patterns with distinctive endpoint letters; rejects many windows early \\[4pt]
 & Smith~\cite{Smith1991}
 & Last text char in window and first text char after window $t[i+m-1]$ and $t[i+m]$
 & Often improves shifts vs.\ Horspool/Sunday with small overhead \\[4pt]
 & Tuned BM~\cite{Hume1991}
 & Last text char in window $t[i{+}m{-}1]$
 & Scans ahead to find a promising window before doing full comparison \\
\midrule
\multirow{2}{=}{\raggedright BC-Only\\char-pair}
 & Zhu--Takaoka~\cite{ZhuTakaoka1987}
 & Last two text chars in window $(t[i{+}m{-}2],\,t[i{+}m{-}1])$
 & Reduces minimal shifts; effective on large alphabets     \\[4pt]
 & Berry--Ravindran~\cite{BerryRavindran1999}
 & First two text chars after window $(t[i{+}m],\,t[i{+}m{+}1])$
 & Large alphabets; often larger shifts; extra preprocessing/space\\
\midrule
BC+GS
 & Boyer--Moore~\cite{Boyer1977}
 & Mismatching text char $t[i{+}j{-}1]$ + matched text suffix
 & General purpose; combines BC and GS for sublinear expected time (random-text model) \\
\bottomrule
\end{tabular}%
}
\end{table}
Table~\ref{tab:taxonomy} summarizes the classical Boyer--Moore variants that we adapt to indeterminate strings, grouped by how they compute a shift. The bad-character-only variants differ only in which text letters they inspect after an alignment. The single-character variants use one letter, the last letter of the window (Horspool, Raita, Tuned BM) or the first letter past it (Sunday), while Raita and Tuned BM additionally use cheap tests on the window before scanning it, and Smith takes the larger of Horspool's and Sunday's shifts. The two character-pair variants look up two adjacent letters at once, the last two letters of the window (Zhu--Takaoka) or the two letters past it (Berry--Ravindran), which gives larger shifts at the cost of a table of size $\sigma^2$. The full Boyer--Moore algorithm combines a bad-character rule with the good-suffix rule and uses the matched suffix as well as the mismatching letter.

\subsection{\textsc{Fast\_BM\_Indet} Algorithm}
For completeness, Algorithm~\ref{alg:fastbmindet} reproduces the \textsc{BM\_Indet}
algorithm of Dehghani et al.~\cite{Neerja2024}, modified only as required by
our new good-suffix rule and new bad-character rule used. Three changes are made. First, during the
right-to-left scan of each window we record $j_1$ and $j_2$, the first and
last solid positions of the matched suffix, and pass them to \textsc{Fast\_GSR\_Indet\_Shift}
(Algorithm~\ref{alg:fast-bmindet-shift}), which recovers the corresponding
solid letters $c_1 = T[i{+}j_1{-}1]$ and $c_2 = T[i{+}j_2{-}1]$ internally.
Second, we replace the \textsc{indet\_gsr\_shift} routine of~\cite{Neerja2024},
which rebuilds a prefix array at every indeterminate alignment, with a call to
\textsc{Fast\_\allowbreak GSR\_\allowbreak Indet\_\allowbreak Shift}. Third we replace the \textsc{bad\_ \allowbreak character\_ \allowbreak rule\_ \allowbreak shift} of \textsc{BM\_Indet} by one of the new BC Rules I-IV (here by BC Rule I). Every other step is unchanged from~\cite{Neerja2024}. Recording
$j_1$ and $j_2$ costs only $O(1)$ per position, so it does not add to the
algorithm's complexity; pairing each of BC~Rules~I--IV with this main
procedure gives the \textsc{Fast\_BM\_Indet} variants analyzed in
Section~\ref{sec:main-results}.

\begin{algorithm}[H]
	\caption{\textsc{Fast\_BM\_Indet} Variant with BC Rule I}\label{alg:fastbmindet}
	\begin{algorithmic}[1]
		\Function{\textsc{Fast\_BM\_Indet}}{$T,n,P,m$}: Integer List
            \State $i \gets 1;\ shift \gets 1;\ mismatched \gets \textsc{false}$
            \State $regular \gets$ ($P$ contains no indeterminate letter)
            \State $\tau' \gets \textsc{preprocess\_bad\_character}(P,m)$
                   \Comment{position-indexed table for the new good-suffix rule}
            \If {$regular$} compute the good-suffix tables $L$ and $\ell$ of $P$ \EndIf
            \State $indexlist \gets \emptyset$ \Comment{indices where $P$ occurs in $T$}
            \While {$i < n-m+1$}
                \State $shift \gets 1;\ mismatched \gets \textsc{false};\ indet_T \gets \textsc{false}$
                \State $j_1 \gets 0;\ j_2 \gets 0$
                       \Comment{first/last solid positions of the matched suffix}
                \For {$j \gets m$ \textbf{downto} $1$}
                    \If {$P[j] \not\approx T[i + j - 1]$}
                        \State $skip\_bc \gets \textsc{BC-RuleI-shift}$
                        \If {\textbf{not} $regular$ \textbf{ or } $indet_T$}
                            \State $skip\_gs \gets \textsc{Fast\_GSR\_Indet\_Shift}(\tau', j_1, j_2)$
                        \Else
                            \State $skip\_gs \gets \textsc{good\_suffix\_rule\_shift}$
                        \EndIf
                        \State $shift \gets \textsc{max}(shift, skip\_bc, skip\_gs)$
                        \State $mismatched \gets \textsc{true}$
                        \State \textbf{break}
                    \EndIf
                    \If {\textsc{Indet}($T[i + j - 1]$)}
                        \State $indet_T \gets \textsc{true}$
                    \Else
                        \Comment{$T[i+j-1]$ is a solid edge letter}
                        \If {$j_2 = 0$}
                            \State $j_2 \gets j$
                        \EndIf
                        \State $j_1 \gets j$
                    \EndIf
                \EndFor
                \If {\textbf{not} $mismatched$}
                    \State $indexlist \gets indexlist \cup \{i\}$
                    \If {\textbf{not} $regular$ \textbf{ or } $indet_T$}
                        \State $skip\_gs \gets \textsc{Fast\_GSR\_Indet\_Shift}(\tau', j_1, j_2)$
                    \Else
                        \State $skip\_gs \gets \textsc{good\_suffix\_rule\_shift}$
                    \EndIf
                    \State $shift \gets \textsc{max}(shift, skip\_gs)$
                \EndIf
                \State $i \gets i + shift$
            \EndWhile
            \State \Return $indexlist$
        \EndFunction
	\end{algorithmic}
\end{algorithm}

\end{document}